\documentclass[runningheads]{llncs}
\usepackage{graphicx}
\usepackage{wrapfig}
\usepackage{float}
\usepackage{amsmath}
\usepackage{amssymb}
\usepackage{mathtools}
\usepackage[english]{babel}
\usepackage{cite}
\usepackage{textcmds}
\usepackage[utf8]{inputenc}
\usepackage[left]{lineno}
\usepackage{dirtytalk}
\usepackage{xspace}
\usepackage{wrapfig}

\newcommand\sbullet[1][.75]{\mathbin{\vcenter{\hbox{\scalebox{#1}{$\bullet$}}}}}

\DeclareMathOperator{\compose}{\sbullet}

\newcommand{\leftinterface}[1]{{^\ast #1}}
\newcommand{\rightinterface}[1]{{#1 ^\ast}}

\DeclareMathOperator{\matches}{\emph{matches}}
\DeclareMathOperator{\matchfree}{\emph{matchfree}}
\DeclareMathOperator{\interior}{\emph{interior}}

\DeclareMathOperator{\sbakersystem}{\emph{baker system}}
\DeclareMathOperator{\svendorsystem}{\emph{vendor system}}

\DeclareMathOperator{\sruns}{\emph{runs}}

\begin{document}
\title{Compositionality of Systems and Partially Ordered Runs}
%
%

\author{Peter Fettke\inst{1,2}\orcidID{0000-0002-0624-4431} \and
Wolfgang Reisig\inst{3}\orcidID{0000-0002-7026-2810}}
\authorrunning{P. Fettke, W. Reisig}
%

\institute{German Research Center for Artificial Intelligence (DFKI), Saarbr\"ucken, Germany \\
\email{peter.fettke@dfki.de}\\ \and
Saarland University, Saarbr\"ucken, Germany \\ \and
Humboldt-Universität zu Berlin, Berlin, Germany \\ 
\email{reisig@informatik.hu-berlin.de}}

\maketitle 
\begin{abstract}
In the late 1970s, C.A. Petri introduced partially ordered event occurrences (runs), then called \emph{processes}, as the appropriate model to describe the individual evolutions of distributed systems. Here, we present a unified framework for handling Petri nets and their runs, specifically to compose and decompose them. It is shown that, for nets $M$ and $N$, the set of runs of the composed net $M \compose N$ equals the composition of the runs of $M$ and $N$.

\end{abstract}

\section{Motivation}\label{sec:0}

In the late 1970s, C.A. Petri proposed partially ordered event occurrences, then called \emph{processes}, as the appropriate model to describe individual evolutions of distributed systems. Their use became known as \say{partial order semantics}, in contrast to \say{interleaving semantics}, which relies on firing sequences. Although the importance of (business) process management is now widely recognized \cite{Weske_19,dumas2018fundamentals,fettke2025bpmai}, the concept of processes as understood by Petri never gained widespread acceptance. We see conceptual and technical obstacles for this. First, Petri's processes order event occurrences not temporally but causally. This is quite unfamiliar to traditionally educated computer scientists. Automata theory, the semantics of programming languages, the behavior of distributed systems, temporal logic, and other theories all organize event occurrences along a temporal scale. This is a strong and idealized assumption that ignores measurement problems in cases of temporally close, independent events. While this assumption works for many system models, it can obscure the understanding of many phenomena. Second, interleaving semantics only handles sequences of states and steps, while processes are technically more complex, especially regarding their composition and refinement.

A fundamental property of system models and their behavior is \emph{compositionality}: the composition of systems must be mirrored by a composition operator on their behavior. Interleaving semantics accomplishes this requirement to some extent by help of -- not too elegant -- \say{shuffle} operators. In this paper, we suggest a general framework for compositionality of partial order semantics. 

In general terms, the principle of \emph{compositionality} assumes two sets, $A$ and $B$, each with internal operators: $\compose_A : A \times A \rightarrow A$ and $\compose_B : B \times B \rightarrow B$, along with a mapping $f: A \rightarrow B$. Compositionality then means that for all $a, b \in A$, it holds that $f(a \compose_A b) = f(a) \compose_B f(b)$. For example, in formal language theory, automata and formal languages are compositional: to compose two finite automata $M$ and $N$, merge the final states of $M$ with the initial state of $N$, resulting in the automaton $M \compose N$. To compose two languages $A$ and $B$, concatenate each word $v$ in $A$ with each word $w$ in $B$, producing the word $vw$. Compositionality in automata and formal languages then signifies that the language $L(M \compose N)$ of the automaton $M \compose N$ equals the composition $L(M) \compose L(N)$ of the languages of $M$ and $N$. In technical terms, using the same symbol for both composition operators:

\begin{equation}\label{eq:1}
L(M \compose N) = L(M) \compose L(N).                                 \end{equation}

For place/transition nets, the situation is more complex: assume, as usual, that each occurrence sequence of a Petri net $N$ contributes a word to the language $L(N)$ of $N$. There are many ways to compose Petri nets. Most non-trivial composition operators $\compose$ imply that, for two Petri nets $M$ and $N$, a language $L(M \compose N)$ contains more words than $L(M) \compose L(N)$. Intuitively, in $M \compose N$, the net $N$ may already start working before $M$ is fully completed. To address this issue, various versions of shuffling operators have been proposed \cite{DBLP:journals/cacm/Gischer81}. 

As an alternative to occurrence sequences, Petri proposed partially ordered runs \cite{Petri_77}. These runs have been studied extensively, but a clear breakthrough is still missing. Two main reasons explain this: conceptually, partially ordered runs replace temporal order with causal order, which requires some revision of classical thinking about system behavior. Technically, notions such as composition and refinement are somewhat complex for partially ordered runs. 

In this contribution, we replace occurrence sequences with a flexible and technically simple concept of partially ordered runs that satisfies equation~\ref{eq:1}. 

The rest of this contribution is organized as follows: The next Section provides an introductory example. Section~\ref{sec:2} discusses the concept of net modules and their composition. Section~\ref{sec:3} introduces steps and runs of net modules, as well as the special case of basic steps and runs. Section~\ref{sec:4} presents the Composition Theorem, which establishes equation~\ref{eq:1}. Section~\ref{sec:5} offers a brief discussion of related work.

\section{An introductory example}\label{sec:1}

To illuminate equation~\ref{eq:1} for the case of Petri nets and their partially ordered runs, Fig.~\ref{fig:1} provides an example, with $M = \sbakersystem$ and $N = \svendorsystem$.

\begin{figure}[t]
   \centering
   \includegraphics[width=\textwidth]{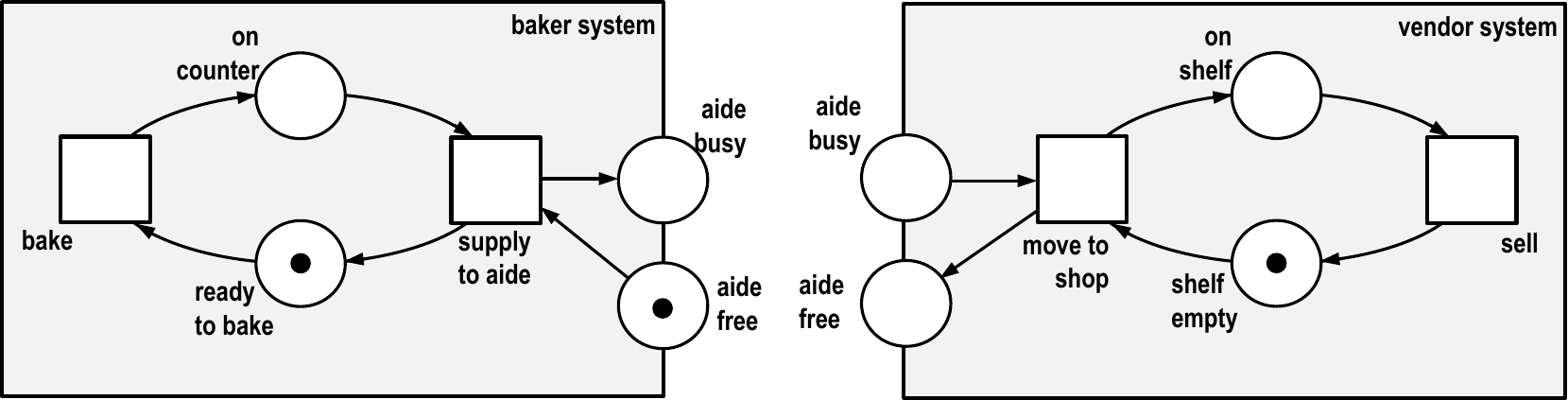}
   \caption{Baker system and vendor system}
   \label{fig:1}
\end{figure}

For the left side of equation~\ref{eq:1}, Fig.~\ref{fig:2} shows the global system, composed of the two systems of Fig.~\ref{fig:1}.

\begin{figure}[t]
   \centering
   \includegraphics[width=\textwidth]{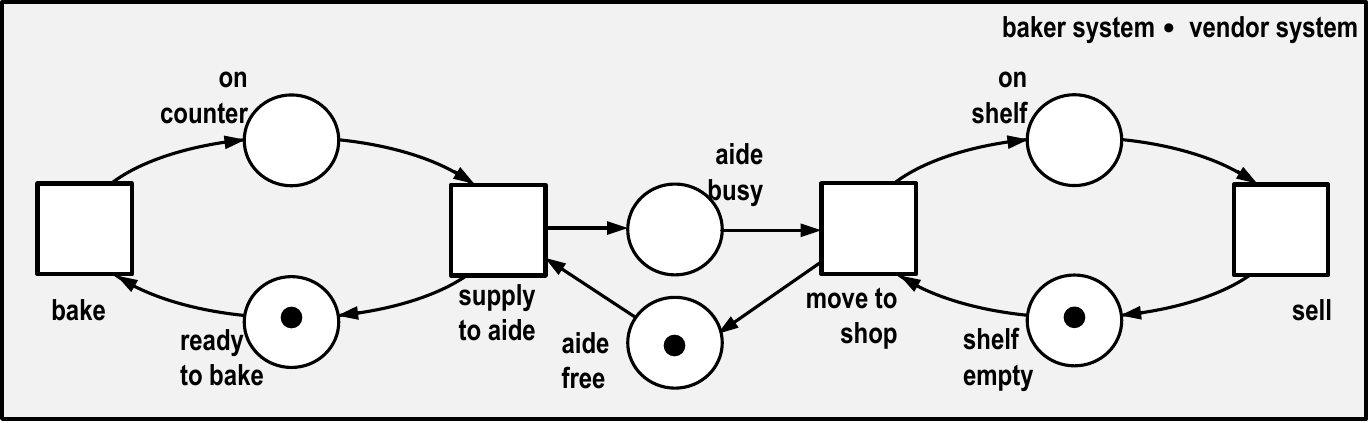}
   \caption{The global system composed of the baker system and the vendor system}
   \label{fig:2}
\end{figure}

\begin{figure}[t]
   \centering
   \includegraphics[width=\textwidth]{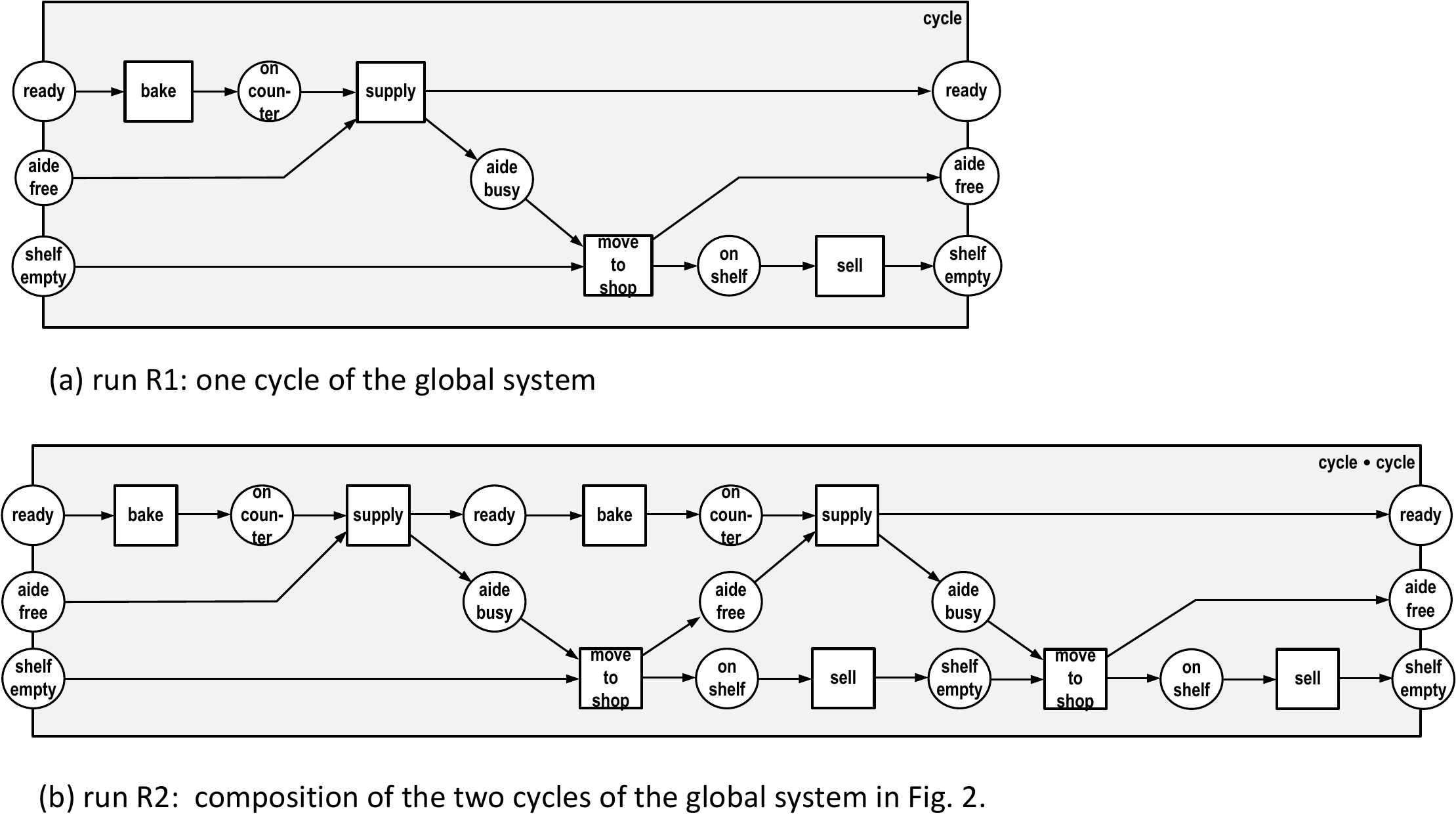}
   \caption{Runs of the global system}
   \label{fig:3}
\end{figure}

Fig.~\ref{fig:3} shows two runs of the global system, composed of the vendor system and the baker system. Fig.~\ref{fig:3}(a), called \emph{cycle}, unfolds a cycle of baking, supplying, moving, and selling a bread. Fig.~\ref{fig:3}(b) shows two intertwined such cycles. Intuitively, a run protocols each single event occurrence as a separate, individual item. For example, in Fig.~\ref{fig:3}, we can distinguish the first occurrence of \emph{bake}, which is not preceded by any event occurrence, from the second occurrence of \emph{bake}, which is preceded by the first occurrence of \emph{supply} -- and the first occurrence of \emph{bake}. Consequently, in technical terms, a sequence of arcs of a run never forms a cycle. And no place has more than one ingoing and one outgoing arc. Just as system modules, also runs can be composed. For example, the run $R_2$ of Fig.~\ref{fig:3}(b) is just the composition of two instances of the run $R_1$ of Fig.~\ref{fig:3}(a). 

Turning to the right-hand side of equation~\ref{eq:1}, we consider runs of the baker system and runs of the vendor system of Fig.~\ref{fig:1}. Fig.~\ref{fig:4} shows those runs.  The upper run displays the composition of two cycles of the baker system; the lower run shows the composition of two runs of the vendor system. Notice that the interfaces of the two runs contain elements with the same label (\emph{aide busy}).

\begin{figure}[t]
   \centering
   \includegraphics[width=\textwidth]{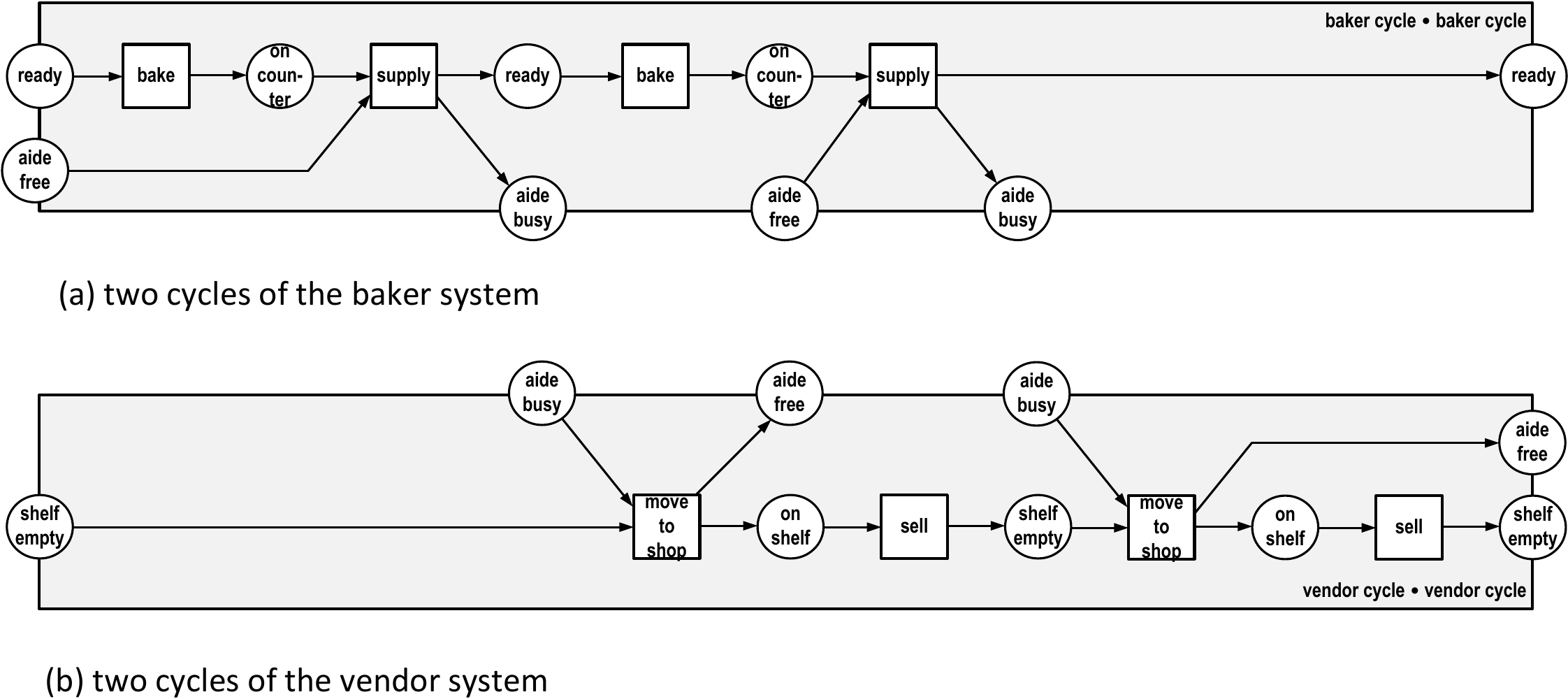}
   \caption{Runs of the baker system and the vendor system}
   \label{fig:4}
\end{figure}

The main point and motivation of the example of this Section is now the observation that the run of the global system, composed from two cycles as in Fig.~\ref{fig:3}(b), coincides with the composition of the two runs of Fig.~\ref{fig:4}! This shows the validity of equation~\ref{eq:1} for both the baker system and the vendor system. 

The systems and runs in Fig.~\ref{fig:1} to \ref{fig:4} come with places on their interfaces, as drawn on the margin of the respective surrounding box. Interfaces may also include transitions, as the systems in Fig.~\ref{fig:5} and the runs in Fig. 6 show. For example, the upmost run \emph{run of take-supply} hast two occurrences of the event \emph{supply to aide} in its right interface. The second run, \emph{run of supply-move}, hast two occurrences of \emph{supply to aide} in its left interface, and two occurrences of \emph{move to shop} in its right interface. Finally, the run \emph{run of move-sell}, has two occurrences of \emph{move to shop} in its right interface.  

\begin{figure}[t]
   \centering
   \includegraphics[width=\textwidth]{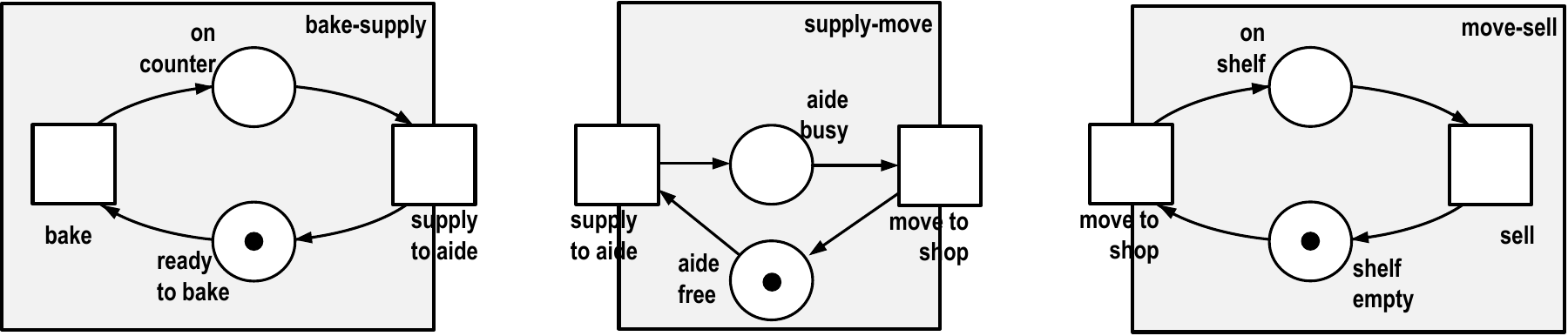}
   \caption{Three system modules}
   \label{fig:5}
\end{figure}

\begin{figure}[t]
   \centering
   \includegraphics[width=\textwidth]{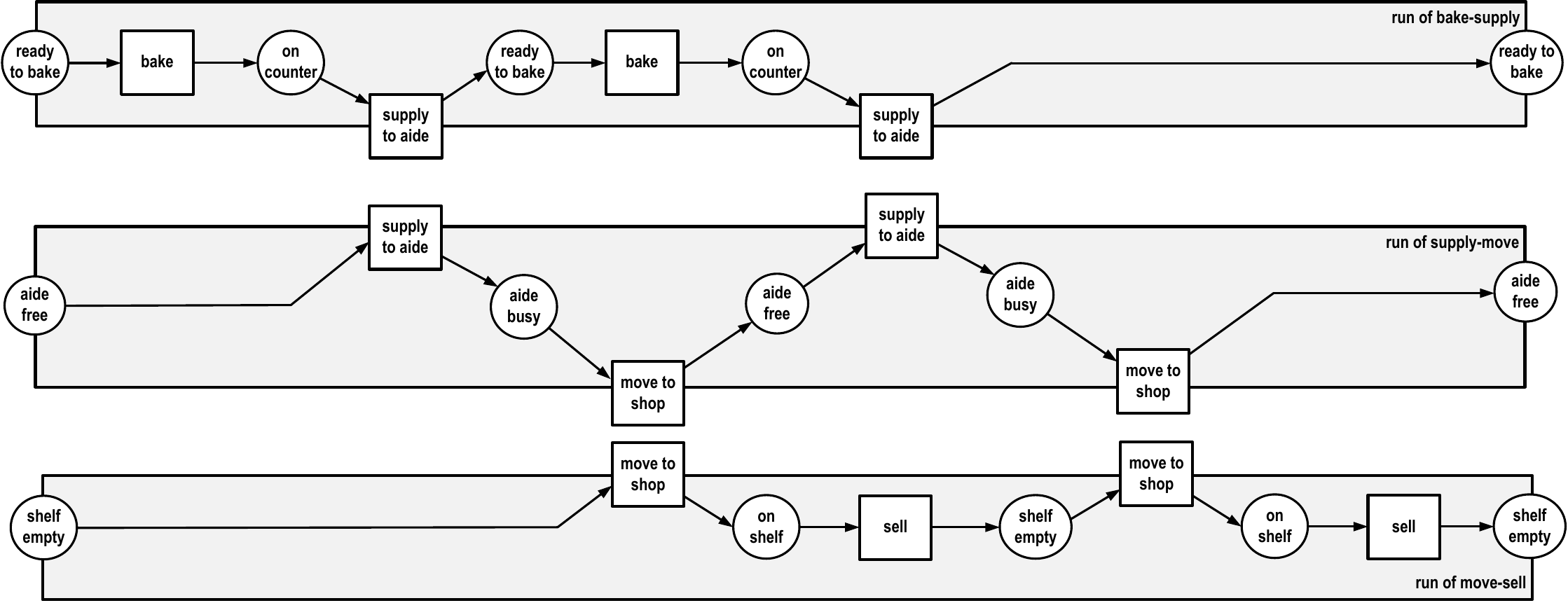}
   \caption{Two cycles of each of the three system modules of Fig.~\ref{fig:5}}
   \label{fig:6}
\end{figure}

The composition of the three runs of Fig.~\ref{fig:6} again yields again the run of Fig.~\ref{fig:3}(b).

\section{Net modules and their composition}\label{sec:2}

At the core of this Section and the entire paper is the concept of a net module: a net with interfaces that allow the module to connect with other modules. Traditional place/transition nets, along with distributed runs, can be transformed into modules. 

To stay self-contained, we review some formal fundamentals about orders and labelings of sets.

\begin{definition}[ordered and labeled set, degree]
\begin{enumerate}

\item A set $A$ is \emph{ordered} if a total and strict order $<$ is assumed.

\item A set $A$ is \emph{labeled} if a set $\Lambda$ of labels and a mapping $\lambda: A \rightarrow \Lambda$ is assumed.

\item The \emph{degree} of an element a of an ordered and labeled set $A$ is the number of elements of $A$ with the same label that are smaller than $a$.
\end{enumerate}
\end{definition}

Later on, we consider the ordered and labeled set $A = \{a, b, c, d\}$, alphabetically ordered, with labeling:

\begin{equation}\label{eq:2}
    \lambda(a) = \textit{ready}, \lambda(b) = \lambda(d) = \textit{aide busy} \text{ and } \lambda(c) = \textit{aide free}.
\end{equation}

The degree of $d$ is $2$; the degree of all other elements of $A$ is $1$. Based on these notions, the concept of interfaces is defined as follows:

\begin{definition}[interface, matching partners, matchfree]
\begin{enumerate}

\item An \emph{interface} is a finite, ordered, and labeled set.

\item For two interfaces $A$ and $B$, two elements $a \in A$ and $b \in B$ are \emph{matching partners} of $A$ and $B$ if a and b are equally labeled and have the same degree. $\matches(A, B)$ denotes the set of all matching partners of $A$ and $B$.

\item $\matchfree(A, B)$ contains all $a \in A$ without a matching partner in $B$. 

\item $\matchfree(B, A)$ contains all $b \in B$ without a matching partner in $A$. 
\end{enumerate}

\end{definition}

Let $B$ be the ordered and labeled set $B = \{e, f, g, h\}$, alphabetically ordered, with labeling:
\begin{equation}\label{eq:3}
\lambda(e) = \textit{shelf empty}, \lambda(f) = \lambda(h) = \textit{aide busy} \text{ and } \lambda(g) = \textit{aide free}.
\end{equation}
  
The degree of $h$ is $2$; the degree of all other elements of $B$ is $1$. Furthermore, the matching partners of the sets $A$ and $B$, as in equations \ref{eq:2} and \ref{eq:3} are $b$ and $f$, $c$ and $g$, as well as $d$ and $h$. $\matchfree(A, B)$ contains $a$, and $\matchfree(B, A)$ contains $e$. 

The sets $\matches(A, B)$, $\matchfree(A, B)$ and $\matchfree(B, A)$ inherit the order and labeling of $A$ and $B$, respectively, and can themselves be conceived as interfaces.

The graph structure of Petri nets is defined as usual:

\begin{definition}[net graph]
Let $P$ and $T$ be finite, disjoint sets. Let $E$ be a relation over $P \cup T$ such that for each tuple $(x, y) \in E$, either, $x \in P$ and $y \in T$, or, $x \in T$ and $y \in P$. Then $M = (P, T; E)$ is a \emph{net graph}. The elements of $P$, $T$, and $E$ are denoted as \emph{places}, \emph{transitions}, and \emph{arcs}.
\end{definition}

Generally, for net graphs, we assume separate sets of labels $\Lambda_P$ and $\Lambda_T$ for places and transitions, respectively. We are now prepared to define \emph{net modules} as net graphs together with left and right interfaces:

\begin{definition}[net module, left and right interfaces, interior of module]
Let $M = (P, T; E)$ be a net graph. Let $\leftinterface{M}$, $\rightinterface{M} \subseteq P \cup T$ be interfaces. Then $M$ together with $\leftinterface{M}$ and $\rightinterface{M}$ is a \emph{net module}. $\leftinterface{M}$ and $\rightinterface{M}$ are the \emph{left} and \emph{right interfaces} of $M$, respectively. $P \cup T  \setminus \leftinterface{M} \cup \rightinterface{M}$ is the \emph{interior} of $M$. 
\end{definition}

\emph{Notation} For technical reasons, we sometimes use the \emph{empty net module} $(\emptyset, \emptyset, \emptyset)$, also written $[\emptyset]$.  

\emph{Graphical Convention} In pictorial representations, net graphs are typically shown in the area of Petri nets, using circles, squares, and arrows. Several additional graphical conventions are used for net modules to visually represent the composition of modules in the most intuitive way.

As a standard convention, a net module $M$ is enclosed in a box, with elements of the left and right interfaces positioned on its respective left and right margins. The identity of the interface elements is usually omitted, the order of the interface elements is arranged vertically, and each element is inscribed by its label. For example, in Fig.~\ref{fig:1}, the right interface of the baker system module and the left interface of the vendor system module both contain two places. The first place is labeled \say{aide busy}, and the second one is labeled \say{aide free}. The left interface of the baker system module and the right interface of the vendor system module are empty. As a second example, Fig.~\ref{fig:3} displays two modules, all of whose left and right interfaces have identical labels.

It is often helpful to draw left interface elements on the top margin and right interface elements on the bottom margin of the enclosing box, as shown in Fig.~\ref{fig:4}. The order proceeds from left to right in both cases and extends the top-down order of the elements on the left and right margins. Fig.~\ref{fig:7} illustrates this convention for the labeled and ordered sets of equations~\ref{eq:2} and \ref{eq:3}. Coincidentally, Fig.~\ref{fig:7} shows abstract versions of the modules in Fig.~\ref{fig:4}.

\begin{figure}[t]
   \centering
   \includegraphics[scale=.45]{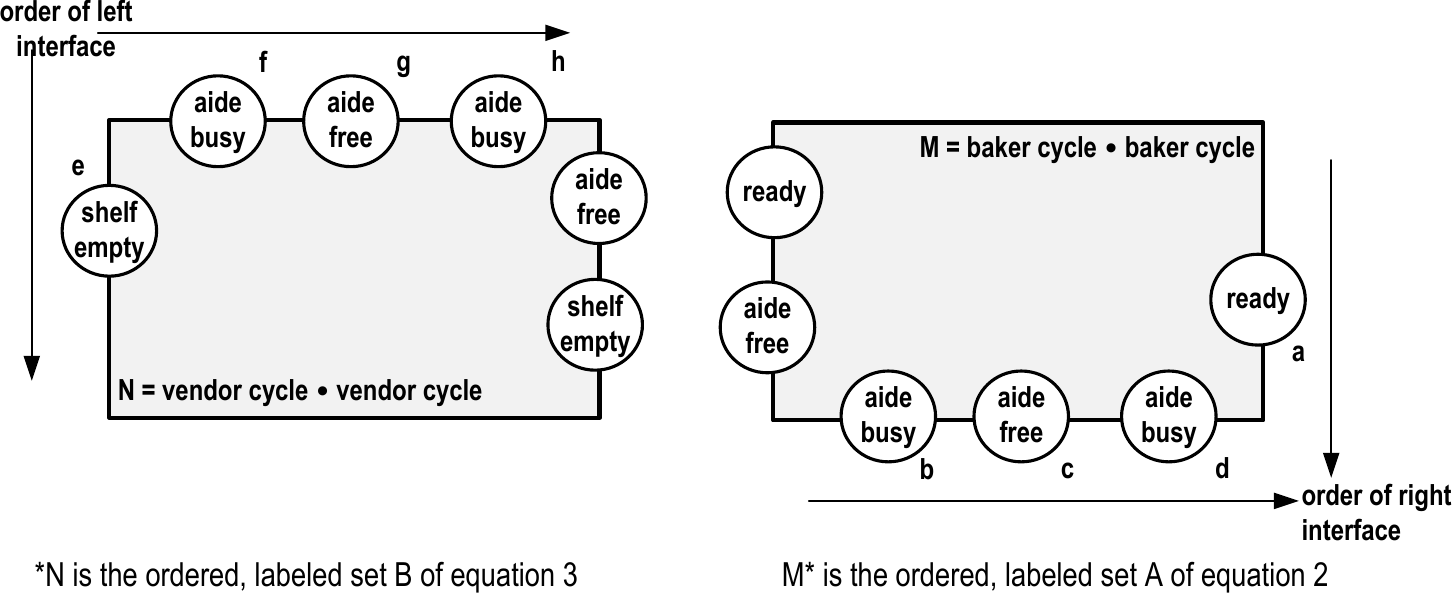}
   \caption{Abstract version of the modules of Fig.~\ref{fig:4}}
   \label{fig:7}
\end{figure}

We are now ready to present the main idea of this paper: the composition of net modules.

\begin{definition}[composition of modules]\label{def:5}
Let $M$ and $N$ be net modules. The \emph{composition of $M$ and $N$} is the net module $M \compose N$, where

\begin{itemize}

\item the $\interior$ of $M \compose N$ is $\interior(M) \cup \interior(N) \cup \matches(\rightinterface{M}, \leftinterface{N})$;

\item the left interface $\leftinterface{(M \compose N)}$ is $\leftinterface{M} \cup \matchfree(\leftinterface{N}, \rightinterface{M})$. The elements of $\leftinterface{M}$ are ordered before the elements of $\matchfree(\leftinterface{N}, \rightinterface{M})$.

\item the right interface $\rightinterface{(M \compose N)}$ is $\rightinterface{N} \cup \matchfree(\rightinterface{M}, \leftinterface{N})$. The elements of $\rightinterface{N}$ are ordered before the elements of $\matchfree(\rightinterface{M}, \leftinterface{N})$.

\item $M \compose N$ inherits the arcs of $M$ and $N$; each match $(a, b)$ inherits the arcs starting or ending at $a$ or at $b$.

\end{itemize}
\end{definition}

\begin{figure}[t]
   \centering
   \includegraphics[scale=.40]{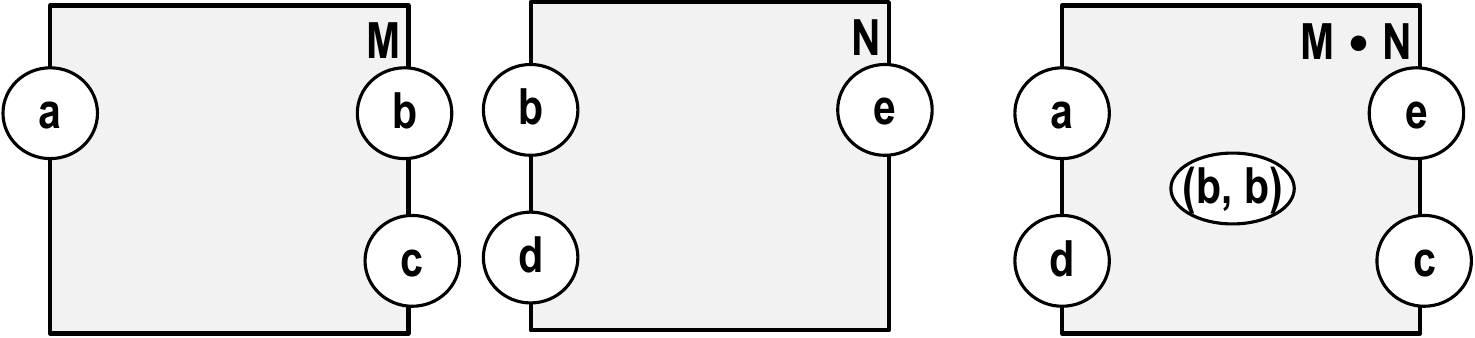}
   \caption{Scetch of the composition operation}
   \label{fig:8}
\end{figure}

\begin{figure}[t]
   \centering
   \includegraphics[width=\textwidth]{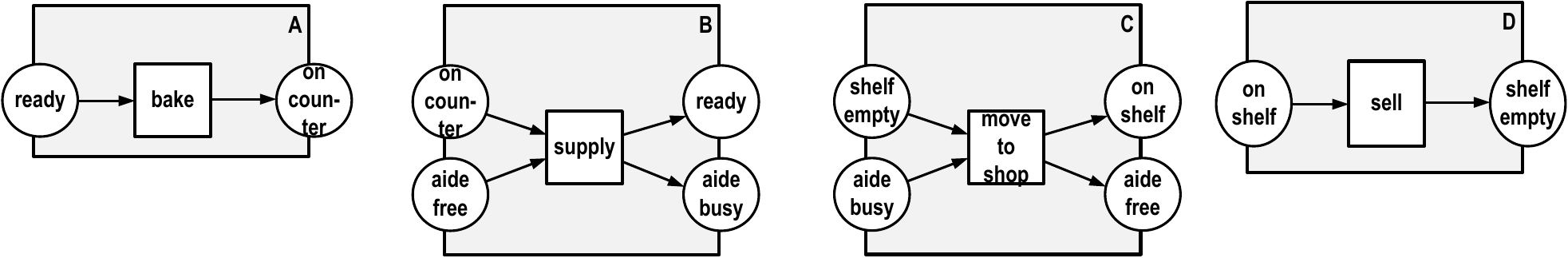}
   \caption{Steps of the net modules of Fig.~\ref{fig:2}}
   \label{fig:9}
\end{figure}

Fig.~\ref{fig:8} illustrates this definition. The interface elements are labeled $a, \dots, e$. The interiors of the modules $M$ and $N$ are not shown. The interior of $M \compose N$ consists of the interior of $M$, the interior of $N$, and the tuple of the $b$-labeled elements of $\rightinterface{M}$ and $\leftinterface{N}$. Clearly, each element of $M$ and $N$ is either an element of $M \compose N$ or contributes to a match of $\rightinterface{M}$ and $\leftinterface{N}$, which then itself is an element of $M \compose N$. The overall system in Fig.~\ref{fig:2} is the composition of the baker system with the vendor system from Fig.~\ref{fig:1}. Fig.~\ref{fig:3}(b) shows the composition of two instances of Fig.~\ref{fig:3}(a). Additionally, Fig.~\ref{fig:3}(b) is also derived as the composition of Fig.~\ref{fig:4}(a) and (b).

Obviously, for the empty net module $[\emptyset]$ holds: $M \compose [\emptyset] = [\emptyset] \compose M = M$.

\section{Runs of net modules}\label{sec:3}

A run of a net module $M$ represents one of its many evolving behaviors. A run consists of \emph{steps}. A step of a transition $t$ in $M$ indicates the occurrence of $t$: the local states that are abandoned and those that are entered. Technically, a step $S$ of $t$ is a labeled net module with a single transition $u$ labeled \say{$t$}. The surrounding places of $u$ in $S$, along with their labels, reflect the corresponding pattern of the surrounding places of $t$ in $M$. Fig.~\ref{fig:9} illustrates steps of the net module from Fig.~\ref{fig:2}. These steps are quite regularly structured: abandoned states form the left interface, entered states form the right interface.

\begin{definition}[step of a net module]\label{def:6}
Let $M$ and $A = (Q, {u}, G)$ be net modules. Let $(p, u)$ be an arc of $A$ if and only if $(\lambda(p), \lambda(u))$ is an arc of $M$, and let $(u, p)$ be an arc of $A$ if and only if $(\lambda(u), \lambda(p))$ is an arc of $M$. Then $A$ is a \emph{step of $M$}.
\end{definition}

Note that this definition does not restrict the interfaces of a step in any way. Any subset of elements of $A$ can serve as an interface of $A$. Therefore, situations are not always as straightforward as shown in Fig.~\ref{fig:9}. In Fig.~\ref{fig:10}, the step $C$ includes the entering state with label \emph{aide free} on its right interface. The step $D$ has the exiting place with label \emph{aide free} on its right interface. Fig.~\ref{fig:11} illustrates the case of transitions within interfaces. 

Matters are not always as simple, as Fig.~\ref{fig:10} shows. Fig.~\ref{fig:11} illustrates the case of transitions in interfaces. 

\begin{figure}[t]
   \centering
   \includegraphics[width=\textwidth]{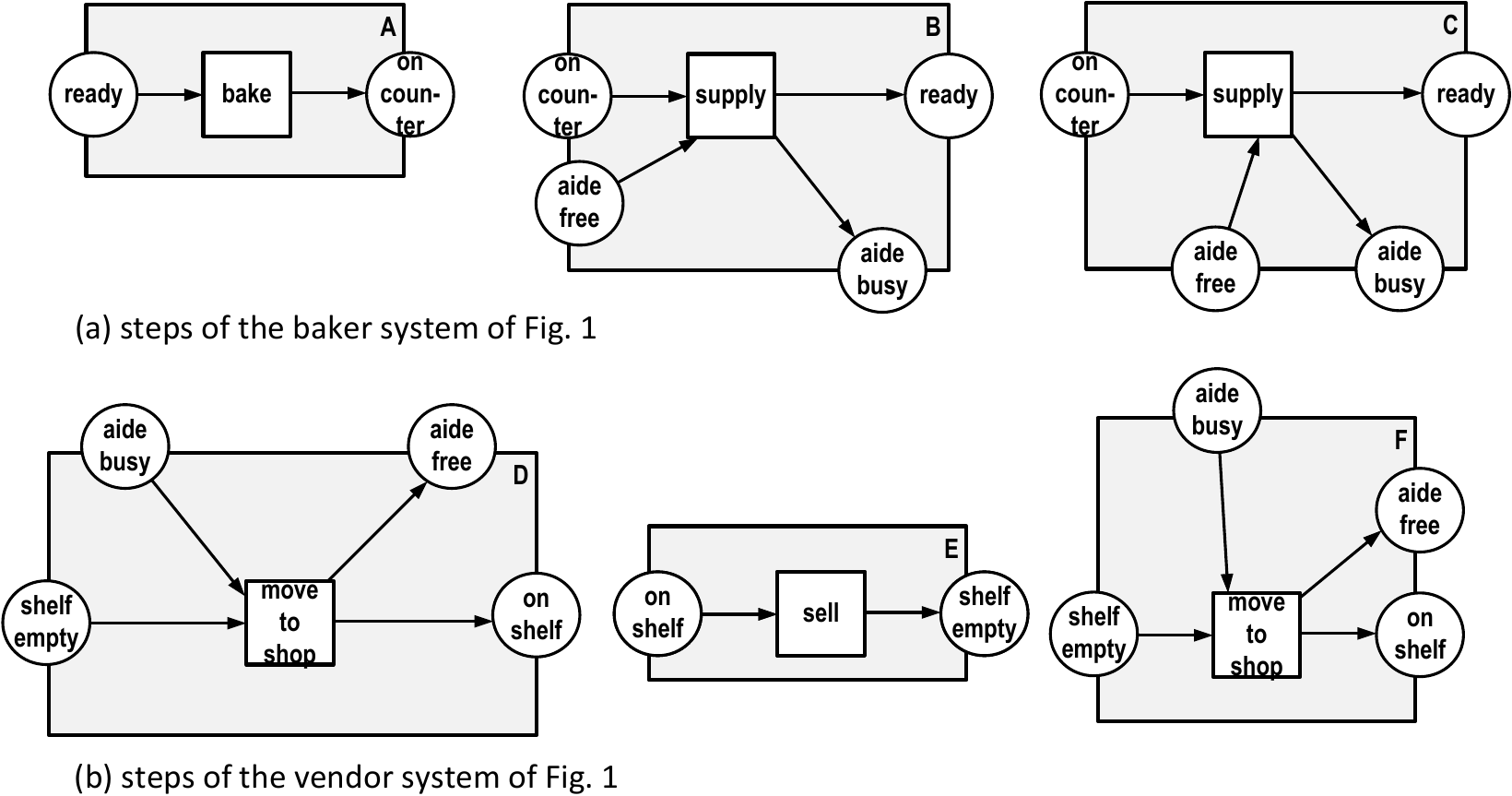}
   \caption{Steps for the runs of the baker system and the vendor system}
   \label{fig:10}
\end{figure}

\begin{figure}[t]
   \centering
   \includegraphics[scale=.45]{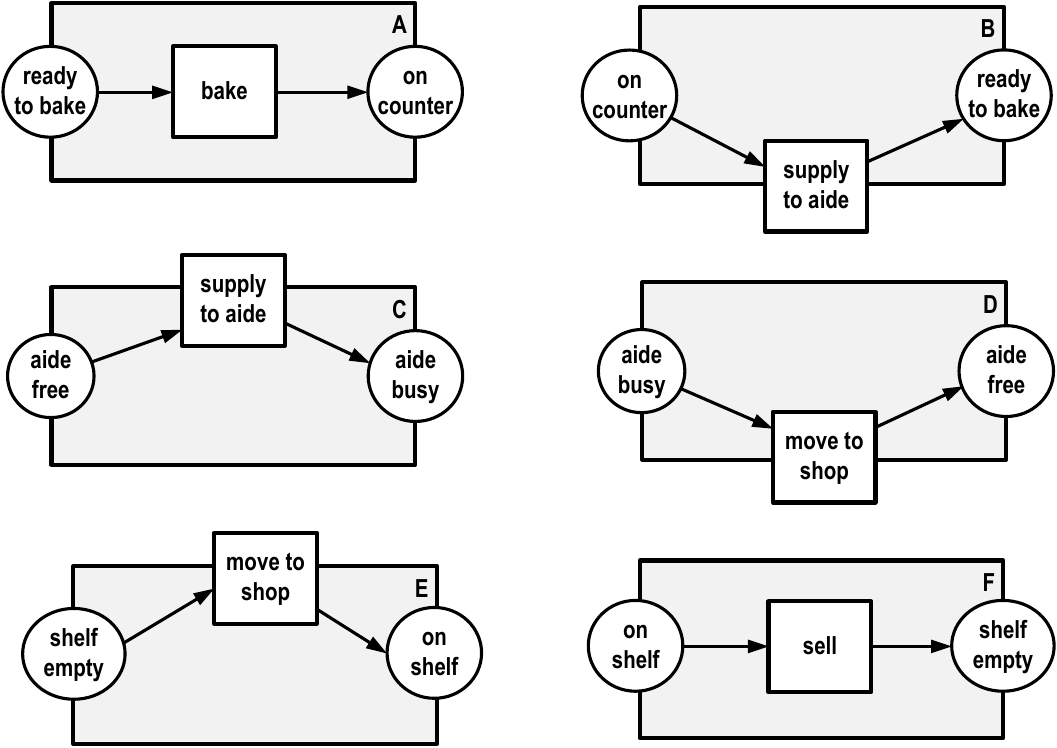}
   \caption{Steps of the three systems modules of Fig.~\ref{fig:5}}
   \label{fig:11}
\end{figure}

Not every acyclic module with unbranched places can be composed from steps: the interior of such a module may require a step order that conflicts with the order generated by the steps on the interface elements. Fig.~\ref{fig:12} shows an example. Therefore, we define runs of net modules inductively as a composition of steps.

A run of a net module $M$ will now formally be defined as composed of steps of $M$. Steps and runs are also net modules; therefore, they are composed as described in Def.~\ref{def:5}. By design, a run is acyclic, and its places are not branched, i.e., at most one arc starts and at most one arc ends at a place. However, the composition of two runs does not necessarily form a run: the result may include a cycle. Even the composition of steps may yield a cycle, as shown in Fig.~\ref{fig:13}.

\begin{figure}[t]
   \centering
   \includegraphics[scale=.45]{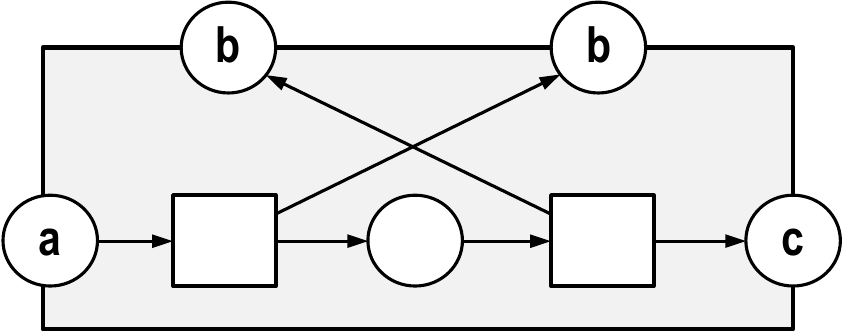}
   \caption{A partially ordered net module with unbranches places}
   \label{fig:12}
\end{figure}

\begin{definition}[runs(M)]\label{def:7}
Let $M$ be a net module. The set $\sruns(M)$ of runs of $M$ is inductively defined:

\begin{enumerate}

\item The empty net module $[\emptyset]$ is a run of $M$.   

\item Let $R$ be a run of $M$ and let $A$ be a step of $N$ such that $R \compose A$ is partially ordered. Then $R \compose A$ is a run of $M$.
\end{enumerate}
\end{definition}

For example, both runs in Fig.~\ref{fig:3} are actually runs of the global system depicted in Fig.~\ref{fig:2}. The upper run, $R_1$, can be formed from the steps in Fig.~\ref{fig:8} simply as $R_1 = A \compose B \compose C \compose D$. The lower run, $R_2$, can then be written as $R_2 = R_1 \compose R_1$. Additionally, the upper run module in Fig.~\ref{fig:4} corresponds to a run of the baker system shown in Fig.~\ref{fig:1}. It can be assembled from the steps in Fig.~\ref{fig:10}(a) as $A \compose B \compose A \compose B$. Similarly, the lower run module in Fig.~\ref{fig:4} is a run of the vendor system in Fig.~\ref{fig:1}, composed of steps from Fig.~\ref{fig:10}(b) as $D \compose E \compose D \compose E$. Lastly, the three runs in Fig.~\ref{fig:6} are created from the steps in Fig.~\ref{fig:11} as $A \compose B \compose A \compose B$, $C \compose D \compose C \compose D$, and $E \compose F \compose E \compose F$.

So far, all examples of systems have been represented as net modules with unbranched places. This applies, in particular, to the net modules in Figures~\ref{fig:1}, \ref{fig:2}, and \ref{fig:5}. The definition of runs as in Def.~\ref{def:7} does not require this limitation. Steps and runs are very well defined also for net modules with branching places. Fig.~\ref{fig:8new} shows an example, extending the system of Fig.~\ref{fig:2} by a second vendor. All runs of the global system of Fig.~\ref{fig:2} are also runs of the system in Fig.~\ref{fig:8new}. 

\begin{figure}[t]
   \centering
   \includegraphics[scale=.40]{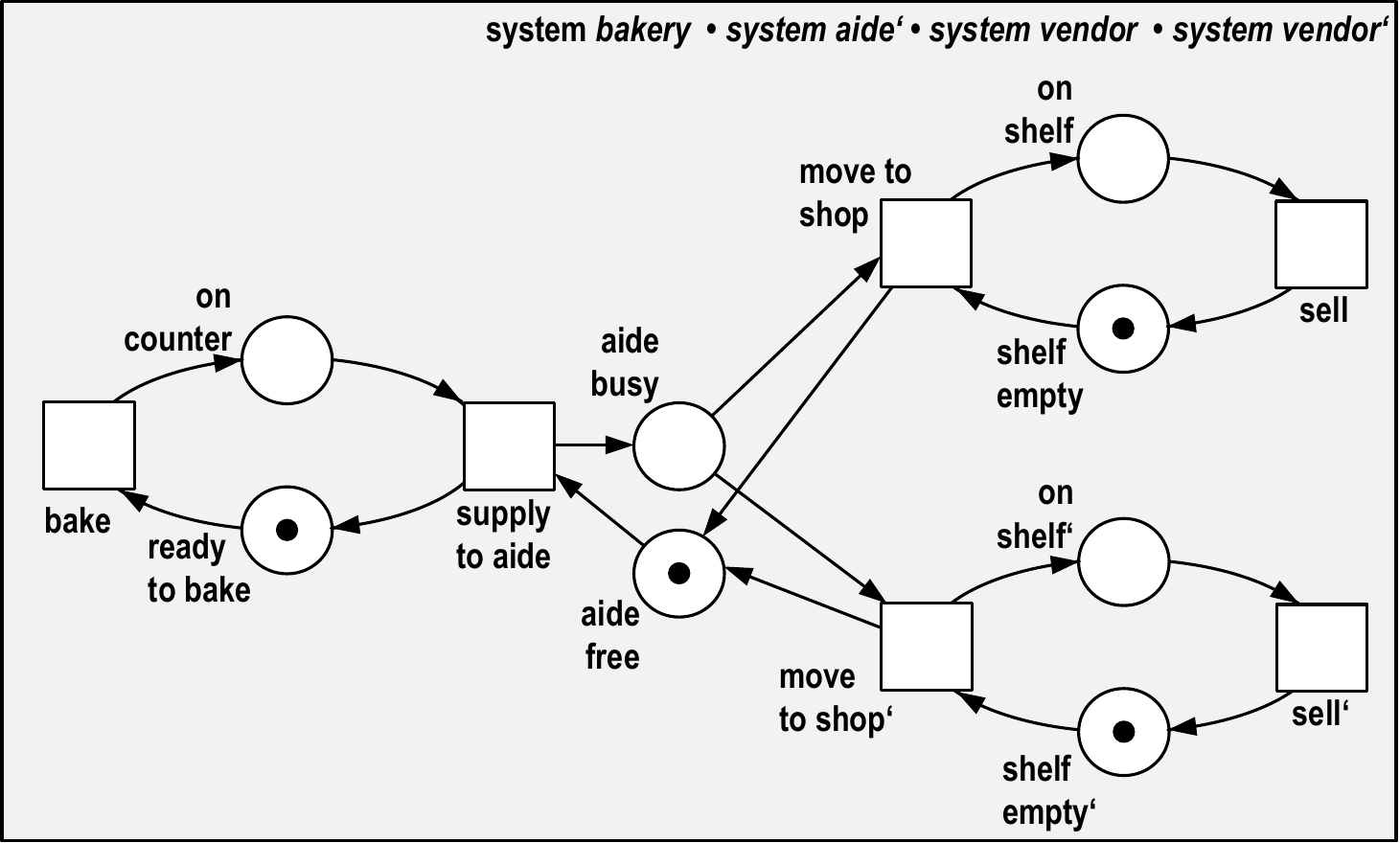}
   \caption{The global system \emph{or} Fig.~\ref{fig:2} with a second seller}
   \label{fig:8new}
\end{figure}

By construction, a run is acyclic, and its places are not branched, i.e., at most one arc starts and at most one arc ends at a place. However, the composition of two runs is not necessarily a run: the result may contain a cycle. Fig.~\ref{fig:13} shows an example.

\begin{figure}[t]
   \centering
   \includegraphics[scale=.40]{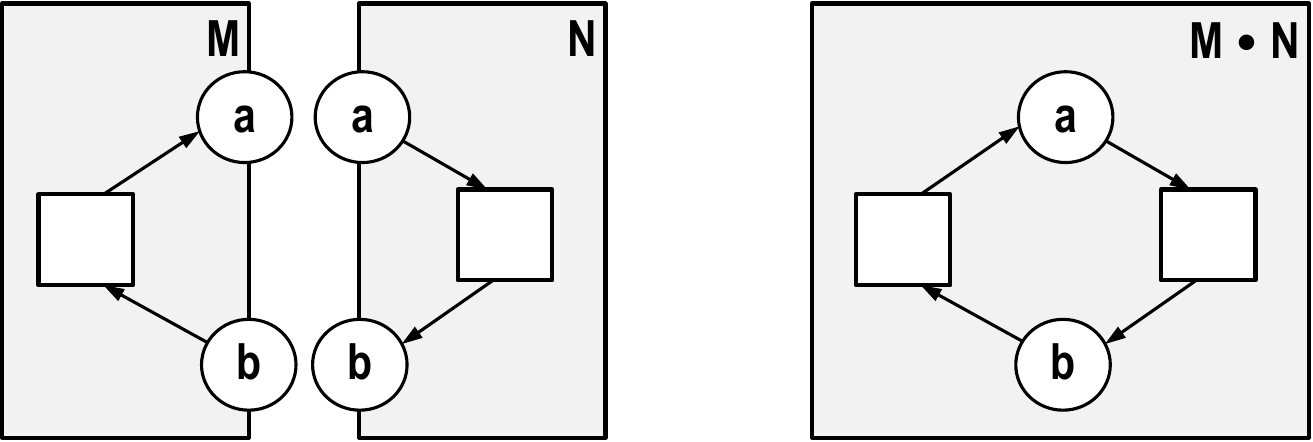}
   \caption{Two technical steps and their composition}
   \label{fig:13}
\end{figure}

In many modeling contexts, the structure of steps is quite regular: ingoing arcs start at the left interface; outgoing arcs end at the right interface. Fig.~\ref{fig:8} shows typical examples of these basic steps. The composition of such steps then produces a \emph{basic run}.

\begin{definition}[basic run]\label{def:8}
A run $R$ is \emph{basic} if $\leftinterface{R}$ consists of all places of $R$ without ingoing arcs, and $\rightinterface{R}$ consists of all places of $R$ without outgoing arcs. 
\end{definition}

Fig.~\ref{fig:3} shows typical basic runs. The realm of basic steps and runs is closed under composition:

\begin{theorem}\label{theorem:1}
The composition $R \compose S$ of basic runs $R$ and $S$ is a basic run. 
\end{theorem}

\begin{proof}
Let $R$ and $S$ be basic runs. By Def.~\ref{def:8}, we have to show:

\begin{itemize}

\item $R \compose S$ is a run, i.e., contains no circle,

\item $\leftinterface{(R \compose S)}$ consists of all places of $R \compose S$ which have no ingoing arcs in $R \compose S$,

\item $\rightinterface{(R \compose S)}$ consists of all places of $R \compose S$ which have no outgoing arcs in $R \compose S$.
\end{itemize}

Fig.~\ref{fig:14} outlines the situation.

To show (i), notice that no element of $\rightinterface{R}$ has an outgoing arc in $R$, and no element of $\leftinterface{S}$ has an ingoing arc in $S$, because $R$ and $S$ are basic runs. Hence, $R \compose S$ contains no arc from $S$ to $R$. Hence, $R \compose S$ contains no circle, as $R$ and $S$ contain no circle.

To show (ii), notice that no element of $\leftinterface{S}$ that contributes to $\leftinterface{R}$ has an ingoing arc in $S$, because $S$ is a basic run. Hence, no element of $\leftinterface{S}$ that contributes to $\leftinterface{(R \compose S)}$ has an ingoing arc in $R \compose S$.

To show (iii), notice that no element of $\rightinterface{R}$ that contributes to $\rightinterface{S}$ has no outgoing arc in $S$, because $R$ is a basic run. Hence, no element of $\rightinterface{R}$ that contributes to $\rightinterface{(R \compose S)}$ has an ingoing arc in $R \compose S$. 

\end{proof}

\begin{figure}[t]
   \centering
   \includegraphics[width=\textwidth]{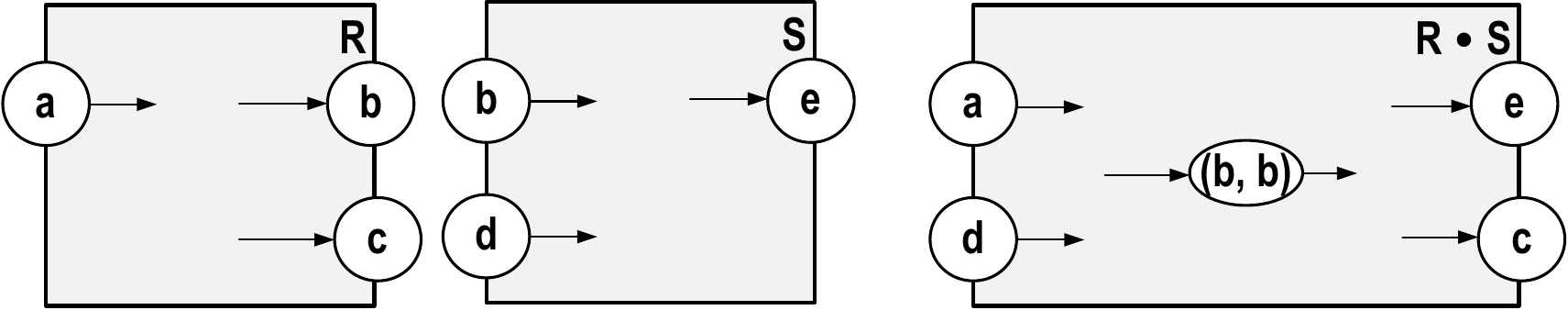}
   \caption{Outline of the proof of Theorem~\ref{theorem:1}}
   \label{fig:14}
\end{figure}

\section{The Composition Theorem }\label{sec:4}

Definition~\ref{def:5} defines the composition of modules, so the left side of equation~\ref{eq:1} in the introduction is well-defined. The right side needs the definition of the composition of sets of runs. To this end, we lift the composition of runs to sets $A$, $B$ of runs in the usual way: $A \compose B$ consists of all composed runs of $A$ and of $B$.

\begin{definition}
Let $A$ and $B$ be two sets of runs. Then $A \compose B = \{a \compose b \text{ } \vert \text{ } a \in A \text{ and } b \in B, \text{ and } a \compose b \text{ is a run}\}$.
\end{definition}

\begin{theorem}
For any two system nets $M$ and $N$ holds:
\begin{center}
$\sruns(M \compose N) = \sruns(M) \compose \sruns(N)$. 
    \end{center}
\end{theorem}

\begin{proof}
    
First, we show: $\sruns(M \compose N) \subseteq \sruns(M) \compose \sruns(N)$.

To be shown: for each run $R$ of $M \compose N$, there exist runs $R_M$ of $M$ and $R_N$ of $N$ such that $R = R_M \compose R_N$.

So, let $R = (P, T; F)$ be a run of $M \compose N$, with labeling function $\lambda$. We show the proposition by induction on $\vert T \vert$. 

\begin{enumerate}

\item Induction basis: Let $\vert T \vert = \emptyset$. Then $R$ is the empty module $[\emptyset]$. $[\emptyset]$ is also a run of $M$ and a run of $N$. Furthermore, obviously holds: $[\emptyset] \compose [\emptyset] = [\emptyset]$.

\item Induction step: Let $\vert T \vert = k$, for $k > 0$. Then $R$ can be written as $R = R_1 \compose \dots \compose R_k$, with steps $R_1, \dots, R_k$ of $M \compose N$ (by Def.~\ref{def:7}). Then there exist runs $R_M$ of $M$ and $R_N$ of $N$ such that $R_1 \compose \dots \compose R_{k-1} = R_M \compose R_N$, by induction assumption. Let $R_k = (P, \{t\}; E)$ and let $\lambda(t) = u$. Then $u$ is either a transition of $M$, or a transition of $N$, or is a match $(u_M, u_N)$ of a transition $u_M$ of $M$ and a transition $u_N$ of $N$ (by Def.~\ref{def:5}). So, we distinguish three cases:

\begin{itemize}

\item Case 1: $u$ is a transition of $M$. Then $R_k$ is a step of $M$ (by Def.~\ref{def:6}). Then $R_k$ is a step of $M \compose N$ (by Def.~\ref{def:5}). Then $R_M \compose R_k$ is a run of $M$. Then $R_M \compose R_k \compose R_N$ is a run of $M \compose N$ (by Def.~\ref{def:7}).

\item Case 2: $u$ is a transition of $N$. Then $R_k$ is a step of $N$ (by Def.~\ref{def:6}). Then $R_k$ is a step of $M \compose N$ (by Def.~\ref{def:5}). Then $R_N \compose R_k$ is a run of $M$ Then $R_M \compose R_N \compose R_k$ is a run of $M \compose N$ (by Def.~\ref{def:7}).

\item Case 3: $u$ is a transition shaped $u = (u_M, u_N)$ with $u_M$ a transition of $M$ and $u_N$ a transition of $N$, and $u_M$ and $u_N$ are matching partners of $\rightinterface{M}$ and $\leftinterface{N}$. Then the step $R_k = (P, \{t\}; E)$ partitions into a step $S_M = (P_M, \{t_M\}, E_M)$ of $M$, and a step $S_N = (P_N, \{t_N\}; E_N)$ of $N$, so that $R_k = S_M \compose S_N$. Then $R_M \compose R_N \compose S_M \compose S_N \compose R_k$ is a run of $M \compose N$ (by Def.~\ref{def:7}). Fig.~\ref{fig:15} outlines this case.
\end{itemize}
\end{enumerate}

\begin{figure}[t]
   \centering
   \includegraphics[scale=.40]{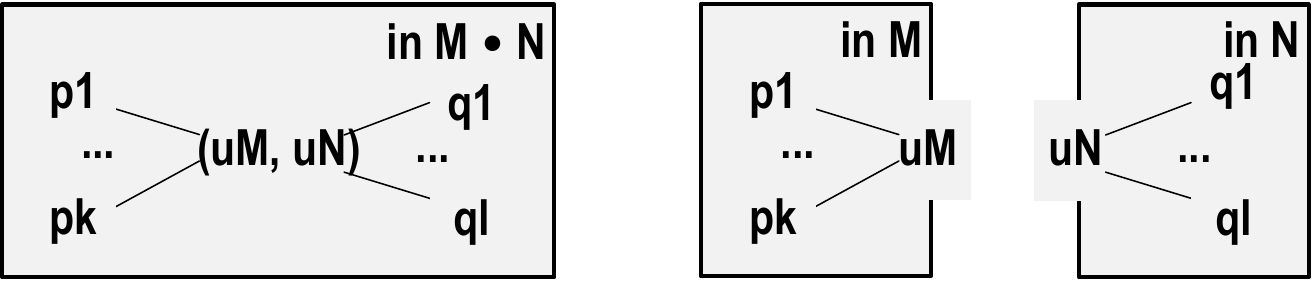}
   \caption{Components of step 3}
   \label{fig:15}
\end{figure}

To show $\sruns(M) \compose \sruns(N) \subseteq \sruns(M \compose N)$, let $R$ be a run of $\sruns(M) \compose \sruns(N)$. We have to show that $R$ is a run of $M \compose N$.

$R$ can be written as $R = A_1 \compose \dots \compose A_l$, with steps $A_1, \dots, A_k$ of $M$ and steps $A_{k+1}, \dots, A_l$, of $N$, for some $1 \leq k \leq l$. Then $A_1, \dots, A_l$ are steps of $M \compose N$ (by Def.~\ref{def:6}). Hence, $R$ is a run of $M \compose N$ (by Def.~\ref{def:7}).
\end{proof}

\section{Related Work}\label{sec:5}

Partial order semantics is often viewed not as a fundamental, conceptual alternative to interleaving semantics, but as a technical trick to speed up analysis algorithms. Specifically, some versions of partial order semantics aim to reduce the \say{state explosion problem} in distributed systems. A typical example in e.g. process mining is \cite{DBLP:conf/apn/SiddiquiAS25}. 

The composition of Petri nets has a long-standing history, with many contributions. We only mention some relevant ideas from the past three decades: \cite{Christensen_Petrucci_93} composes many Petri net modules in one go, merging equally labeled places as well as transitions. Associativity of composition is not addressed, but it seems obvious that this composition operator is not associative. A long-standing initiative with many variants is algebraic calculi for Petri nets, such as the box calculus and Petri net algebras in various forms \cite{Best_et_al_01}. This text collects properties of infinite partial-order runs. In the spirit of process algebras, these calculi define classes of nets inductively using various composition operators, merging transitions with the same labels. The associativity of composition is mostly assumed but rarely explicitly discussed. \cite{Baldan_08} defines composition $A \compose B$ of \say{composable} nets $A$ and $B$ within a categorical framework, where $A$ and $B$ share a common net fragment, $C$. The effect of composition on distributed runs is studied, but the question of whether composition is associative is not addressed. 

\cite{Kindler_Petrucci_09} suggests a general framework for \emph{modular} Petri nets, i.e., nets with features to compose a net with its environment. A module may occur in many instances. Associativity of composition is not discussed in this paper, but is implicitly assumed. A recent survey on composition operators for Petri nets is \cite{DBLP:conf/apn/AmparoreD22}. 

Petri nets with two-faced interfaces are discussed in \cite{Rathke_et_al_14} and \cite{Sobocinski_10}. They introduce Petri nets with boundaries (PNB) that feature two composition operators. Interestingly, both operators are associative. 

Double-faced interfaces for modeling techniques beyond Petri nets are present in the literature, though they tend to be more specialized than the version shown here. A common example is the \say{piping} or \say{chaining} operator $P >> Q$ used in the language CSP \cite{Roscoe_10}. The associativity of $>>$ is assumed without further discussion.

Few contributions combine runs and modularization. Early and detailed studies include \cite{Vogeler92} and \cite{LindeGoers94}. \cite{Kindler_97} proposes a compositional partial order semantics. This composition operator simply fuses equally labeled interface elements. However, it is not associative. 

The initial ideas for modules and the composition operator used in this paper originate from runs of Petri nets described in \cite{Reisig_05}. Associativity of the composition operator was proven in \cite{Reisig_19}. For more details, see \cite{Fettke_Reisig_24}. Our runs form a monoid, similar to how formal languages do. The main theorem of this paper connects systems and their runs in a compositional way. To our knowledge, our composition operator, featuring double-faced interfaces, associative composition, and compositionality, is unmatched.

%
%
%

\bibliographystyle{splncs04}
\bibliography{main}

@inproceedings{DBLP:conf/apn/AmparoreD22,
  author       = {Elvio G. Amparore and
                  Susanna Donatelli},
  editor       = {Luca Bernardinello and
                  Laure Petrucci},
  title        = {The Ins and Outs of {P}etri Net Composition},
  booktitle    = {Application and Theory of Petri Nets and Concurrency - 43rd International Conference, {PETRI} {NETS} 2022, Bergen, Norway, June 19-24, 2022, Proceedings},
  series       = {Lecture Notes in Computer Science},
  volume       = {13288},
  pages        = {278--299},
  publisher    = {Springer},
  year         = {2022}
}

@article{DBLP:journals/cacm/Gischer81,
  author       = {Jay L. Gischer},
  title        = {Shuffle Languages, Petri Nets, and Context-Sensitive Grammars},
  journal      = {Commun. {ACM}},
  volume       = {24},
  number       = {9},
  pages        = {597--605},
  year         = {1981},
  doi          = {10.1145/358746.358767}
}

@book{LindeGoers94,
    author = {Hans-Günther Linde-Göers},
    title = {Compositional Partial Order Semantics of Petri Boxes},
    publisher = {Shaker},
    year = {1994}
}

@inproceedings{DBLP:conf/apn/SiddiquiAS25,
  author       = {Ariba Siddiqui and
                  Wil M. P. van der Aalst and
                  Daniel Schuster},
  editor       = {Elvio Gilberto Amparore and
                  Lukasz Mikulski},
  title        = {Computing Alignments for Partially-Ordered Trace Through Petri Net Unfoldings},
  booktitle    = {Application and Theory of Petri Nets and Concurrency - 46th International
                  Conference, {PETRI} {NETS} 2025, Paris, France, June 22-27, 2025,
                  Proceedings},
  series       = {Lecture Notes in Computer Science},
  volume       = {15714},
  pages        = {411--432},
  publisher    = {Springer},
  year         = {2025}
}

@book{Vogeler92,
    author = {Walter Vogeler},
    title = {Modular Construction and Partial Order Semantics of Petri Nets},
    publisher = {Springer},
  series       = {Lecture Notes in Computer Science},
  volume       = {625},
    year = {1992}
}

@book{dumas2018fundamentals,
    author = {Marlon Dumas and Marcello La Rosa and Jan Mendling and Hajo A. Reijers},
    title = {Fundamentals of Business Process Management},
  publisher    = {Springer},
    edition = 2,
    year = 2018
}

@article{fettke2025bpmai,
	author = {Fettke, Peter and Di Francescomarino, Chiara},
	journal = {KI - K{\"u}nstliche Intelligenz},
	number = {2},
	pages = {67--79},
	title = {Business Process Management and Artificial Intelligence},
	volume = {39},
	year = {2025}
}

@book{Fettke_Reisig_24,
 author = {Fettke, Peter and Reisig, Wolfgang},
 publisher = {Springer},
 title = {Understanding the Digital World: Modeling with \textsc{Heraklit}},
 year = {2024}
}

@book{Best_et_al_01,
 author = {Best, Eike and Devillers, Raymond and Koutny, Maciej},
 publisher = {Springer Science \& Business Media},
 title = {Petri Net Algebra},
 year = {2001}
}

@inproceedings{Christensen_Petrucci_93,
 author = {Christensen, S{\o}ren and Petrucci, Laure},
 booktitle = {International Conference on Application and Theory of {P}etri Nets},
 organization = {Springer},
 pages = {113--133},
 title = {Towards a Modular Analysis of Coloured {P}etri Nets},
 year = {1993}
}

@inproceedings{Kindler_97,
 author = {Kindler, Ekkart},
 booktitle = {Application and Theory of Petri Nets 1997: 18th International Conference, ICATPN'97 Toulouse, France, June 23--27, 1997 Proceedings 18},
 organization = {Springer},
 pages = {235--252},
 title = {A Compositional Partial Order Semantics for {P}etri Net Components},
 year = {1997}
}

@inproceedings{Kindler_Petrucci_09,
 author = {Kindler, Ekkart and Petrucci, Laure},
 booktitle = {Applications and Theory of Petri Nets: 30th International Conference, PETRI NETS 2009, Paris, France, June 22-26, 2009. Proceedings 30},
 organization = {Springer},
 pages = {43--62},
 title = {Towards a Standard for Modular {P}etri Nets: a Formalisation},
 year = {2009}
}

@article{Petri_77,
 author = {Petri, Carl-Adam},
 journal = {Gesellschaft f\"ur Mathematik und Datenverarbeitung},
 title = {{N}on-{S}equential {P}rocesses, {I}nterner {B}ericht Isf-77-5},
 year = {1977}
}

@inproceedings{Rathke_et_al_14,
 author = {Rathke, Julian and Soboci{\'n}ski, Pawe{\l} and Stephens, Owen},
 booktitle = {Reachability Problems: 8th International Workshop, RP 2014, Oxford, UK, September 22-24, 2014. Proceedings 8},
 organization = {Springer},
 pages = {230--243},
 title = {Compositional Reachability in {P}etri Nets},
 year = {2014}
}

@inproceedings{Reisig_05,
 author = {Reisig, Wolfgang},
 booktitle = {International Conference on Application and Theory of Petri Nets},
 organization = {Springer},
 pages = {349--364},
 title = {On the Expressive Power of {P}etri Net Schemata},
 year = {2005}
}

@article{Reisig_19,
 author = {Reisig, Wolfgang},
 journal = {Acta Informatica},
 number = {3},
 pages = {229--253},
 publisher = {Springer},
 title = {Associative Composition of Components with Double-Sided Interfaces},
 volume = {56},
 year = {2019}
}

@book{Roscoe_10,
 author = {Roscoe, Andrew W},
 publisher = {Springer Science \& Business Media},
 title = {Understanding Concurrent Systems},
 year = {2010}
}

@inproceedings{Sobocinski_10,
 author = {Soboci{\'n}ski, Pawe{\l}},
 booktitle = {International Conference on Concurrency Theory},
 organization = {Springer},
 pages = {554--568},
 title = {Representations of {P}etri Net Interactions},
 year = {2010}
}

@book{Weske_19,
 author = {Weske, Mathias},
 edition = {3rd},
 journal = {Berlin},
 title = {Business Process Management: Concepts, Languages, Architectures},
 publisher = {Springer Berlin, Heidelberg},
 year = {2019}
}

@inproceedings{Baldan_08,
  author       = {Paolo Baldan and
                  Andrea Corradini and
                  Hartmut Ehrig and
                  Barbara K{\"{o}}nig},
  editor       = {Hartmut Ehrig and
                  Reiko Heckel and
                  Grzegorz Rozenberg and
                  Gabriele Taentzer},
  title        = {Open {P}etri Nets: Non-deterministic Processes and Compositionality},
  booktitle    = {Graph Transformations, 4th International Conference, {ICGT} 2008,
                  Leicester, United Kingdom, September 7-13, 2008. Proceedings},
  pages        = {257--273},
  publisher    = {Springer},
  year         = {2008}
}





\end{document}